\documentclass[a4paper,11pt]{article}
\usepackage[utf8]{inputenc}
\usepackage{amsmath, bm}
\usepackage{amsthm}
\usepackage{amsfonts}
\usepackage{amssymb}
\usepackage{graphicx, color, subfigure, epstopdf}
\usepackage{breqn}
\usepackage{booktabs,adjustbox}
\usepackage{float}
\usepackage{mathtools}
\usepackage{titlesec}
\usepackage{enumitem}
\usepackage{rotating}
\usepackage{atbegshi,picture}
\usepackage{chngpage}
\usepackage{algorithm, algpseudocode}
\titleformat*{\subsection}{\normalfont}

\newtheorem{mydef}{Definition}[section]
\newtheorem{ex}{Example}[section]
\newtheorem{coro}{Corollary}[section]
\newtheorem{theorem}{Theorem}[section]
\newtheorem{prop}{Proposition}[section]
\newtheorem{remark}[theorem]{Remark}

\title{Mixed models as an alternative to Farima.}
\author{Jos\'e Igor Morlanes}

\begin{document}

\maketitle

\begin{abstract}
We construct a new process using a fractional Brownian motion and 
a fractional Ornstein-Uhlenbeck process of the Second Kind as building blocks. We consider the increments of the new process
in discrete time and, as a result, we obtain a more parsimonious process with similar autocovariance structure to that of a FARIMA. In practice,
variance of the new increment process is a closed-form expression easier to compute than that of FARIMA.
\end{abstract}

\section{Introduction}
 
Models, such as FARIMA or fractional exponential process (FEXP) may be adequated for modeling 
long and short dependencies observed in financial data \cite{Willinger,Lo}. In our paper, we introduce another process in discrete time, the mixed fractional 
Gaussian noise (mfGn) with similar autocovariance structure to the previous ones, i.e, its autocovariance function captures long and short correlations.

There is two main reasons for doing this. The first reason is to reduce the model risk introduced by incorrect calibration, i.e., parameters might be estimated 
with error, they may not be kept up-to-date, and so forth. Both models capture the short and long range dependencies. However a mfGn is more 
parsimonious model than a FARIMA one since for the former, we have to estimate only three parameters, Hurst and gamma parameters and the variance, 
but for the latter we have the AR and MA lag polinomials, the fractional integrated parameter and the variance, with a total of p+q+2 parameter. 

The second reason is that, even theoretically, the autocovariance of a FARIMA process is well-known; it seems extremely difficult to implement 
computationally \cite{Doornik}. On the other hand, we want to implement a new model with an easy close-form expression for the autocovariance function. 
Its discretaze version is more parsimonious and easy to compute with the consequent reduction of numerical errors involving calculations. 
For example, in the calculation of the risk of a position or the pricing of a financial instrument.

Finally, a discrete model depends on the time aggregation or systematic sampling. For example, 
if we assume a FARIMA process, follows a model of the type

\begin{equation*}
 \Phi(L)y_t=\Theta(L)\varepsilon_n
\end{equation*}

where $t = 0, 1, 2, \ldots$, $\Phi(L)$ and $\Theta(L)$ are lag polynomials and $\varepsilon_n$ is an error term. 
Conversely, the temporally aggregated series, $Y_T$, follows the model

\begin{equation*}
\beta(L)Y_T=\xi(L)\epsilon_n
\end{equation*}

where $T = 0, k, 2k,\ldots$, $\beta(L)$ and $\xi(L)$ are aggregate lag polynomials and the operator $L$ is in T time
units, running in kt periods. The variable $\epsilon$ is an error term. In the case of a mfGn, the continuous time model,
is not affected by the sampling frequency.

\section{Fractional Autoregressive Integrated Moving Average}

A time series $\widetilde{X}_n$ is said to be a FARIMA(p, d, q) process if it follows the equation

\begin{equation}
 \Phi_p(L)(1-L)^d(\widetilde{X}_n-\mu)=\Theta_q(L)\varepsilon_n
\end{equation}

where $\varepsilon_n$ is a sequence of i.i.d gaussian random variables. 

Let $\Delta$ be the difference operator defined as $\Delta X_n=X_n-X_{n-1}$. Then 
the difference series $\Delta^d(\widetilde{X}_n -\mu)=(1-L)^d(\widetilde{X}_{n}-\mu)$ follows a stationary and invertible ARMA(p,q) model with 
$L$ the lag operator, $d \in(-\frac{1}{2}, \frac{1}{2})$ the fractional integration parameter, and the AR polynomial, and the MA polynomial respectively 
given by

\begin{align*}
\Phi_p(L)  &= 1-\phi_1L\ldots-\theta_pL^p   \\
 \Theta_q(L) &= 1-\theta_1L-\ldots-\theta_qL^q 
\end{align*}

the AR polynomial, and the MA polynomial respectively.

The model has strong memory because the $\theta_i$ coefficients in its MA representation (\ref{MA}) do not decay over time to zero, implying that the
past shock $\varepsilon_i$ of the model has a permanent effect on the series.

\subsection{Evaluation of FARIMA autocovariance function}

As mention already in the introduction, the autocovariance of a FARIMA process is extremely difficult to implement 
computationally. For example, a very simple procedure is to compute the autocovariances from the MA representation. 

\begin{equation}\label{MA}
 Z_n=\Phi_p(L)^{-1}(1-L)^{-d}\Theta_q(L)\varepsilon_n=\sum_{n=0}^{\infty}\phi_n^z\varepsilon_n
\end{equation}

with $\psi_0=1$.

Then, the autocovariance of a FARIMA process is:

\begin{equation}
 \gamma_k=\sum_{j=0}^{\infty} \psi_j^z\psi_{j+|k|}^z\sigma_{\varepsilon}^2
\end{equation}
The drawback is that, because $\psi_j$ declines hyperbolically, many terms are needed for an accurate approximation.

A seemingly simple alternative is to numerically integrate over the spectrum:
\begin{equation}
 \gamma_k=\int_{-\pi}^{\pi} f_z(\omega)e^{iz\omega}\mathrm{d}\omega
\end{equation}

where the spectrum of the FARIMA process, $f_z(\omega)$, is easily computed. However, numerical integration for each k does rapidly get 
prohibitively slow.

A computacionally optimal autocovariance function of a FARIMA process for implementation is:
\begin{equation}
\label{autocovariance_farima}
 \gamma_i=\sigma^2_{\varepsilon} \sum_{k=-q}^{q}\sum_{j=1}^{p}\psi_k\zeta_jC(d,p+k-i,\rho_j),
\end{equation}

where $\rho_1,\ldots,\rho_p$ are the roots (possibly complex) of the AR polynomial, and 
\begin{equation*}
 \psi_k=\sum_{s=|k|}^{q} \theta_s \theta_{s-|k|},\;\;\zeta^{-1}_j=\rho_j\left[\prod_{i=1}^p(1-\rho_i\rho_j)
\prod_{m=1 m\neq j}^p(\rho_j-\rho_m)\right]
\end{equation*}
where  $\theta_0=1$. C is defined as
\begin{eqnarray*}
C(d,h,\rho)&=&\frac{\Gamma(1-2d)}{[\Gamma(1-d)]^2}\frac{(d)_h}{(1-d)_h}\\
&\times& \left[ \rho^{2p}F(d+h,1;1-d+h;\rho) + F(d-h,1;1-d-h;\rho)-1\right]
\end{eqnarray*}

Here $\Gamma$ is the gamma function, $\rho_j$ are the roots of the AR polynomial, 
and $F(a,1;c,\rho)$ is the hypergeometric function. See more technical details in \cite{Sowell} and \cite{Doornik}.

\section{Mixed Fractional Gaussian process}
  
\subsection{Auxiliary Processes}

In this section, we introduce the processes use below. We follow mainly \cite{Taqqu}, \cite{Mishura}, and \cite{Karatzas}.
We consider, throughout, some underlying complete probability space $(\Omega, \mathcal{F},\mathbb{P})$ and 
denote by $\mathcal{F}_t$ the sigma field representing the publicly available information at time t. Typically, 
$\mathcal{F}_t = \sigma(X_s : s\leq t)$, the sigma field generated by past and present values of the process in question $X$, 
often called the history, up to and including time t.


\subsubsection{Fractional Gaussian Noise}
To capture the long range dependence in the data we use a fractional Gaussian noise (fGn).
First, we define fractional Brownian motion:
\label{fBm}
The fractional Brownian motion (fBm) with Hurst parameter $H \in (0,1)$ is a Gaussian process $B^H=\{B^H_t, t\in \mathbb{R}\}$ 
having the properties
\begin{itemize}
\item[(i)] $B_0^H=0$,
\item[(ii)] $\mathbb{E}B_t^H=0, t\in \mathbb{R}$,
\item[(iii)] $\mathbb{E}B_t^H B^H_s = \frac{1}{2}(|t|^{2H}+|s|^{2H}-|t-s|^{2H}), s,t \in \mathbb{R}$.
\item[(iv)] In the special case of $H=\frac{1}{2}$. $W$ denotes a standard Brownian motion with independent increments.
\end{itemize}
  
If $B^H$ is fBm, then the increment sequence $Z^H_k=B^H_{k+1}-B^H_k$ for $k\in\mathbb{Z}$ is called
fractional Gaussian noise.

\begin{prop}
The process $Z^H$ has the following properties
\begin{enumerate}
\item $Z^H$ is stationary,
\item $\mathbb{E}Z^H_k=0$,
\item $\mathbb{E}(Z^H_k)^2=\sigma^2=\mathbb{E}(Z^H_1)^2$
\item The autocovarinace function of the process $Z^H$ is given by \\[0.2cm]
 $\gamma_k= \frac{\sigma^2}{2}(|k+1|^{2H}-2|k|^{2H}+|k-1|^{2H})$\\
\item If $\frac{1}{2}<H<1$ then $Z^H$ has long range dependence and $\gamma_k>0$.
 \end{enumerate}
\end{prop}

\subsubsection{Fractional Ornstein-Uhlenbeck process of the Second Kind}
To capture the short range dependence, we use a process proposed by Kaarakka and Salminen \cite{Paavo}. 
Let $B^H=\{ B^H_t : t\geq 0\}$ be a fractional Brownian motion with self-similarity
parameter $H \in (0,1)$ with the properties above.

We derive a new Gaussian process by means of Doob's transform of $B^H$:
\begin{equation}
 X_t^{(D,\alpha)} := e^{-\alpha t}B_{a_t}^H,\;\;t\in R
\end{equation}
where $\alpha>0$ and $a_t := a(t,H):= He^{\alpha t /H}/\alpha$. 

The covariance function of $Xt^{(D,\alpha)}$ can be computed from definition (\ref{fBm}) point 4. For $t>s$ we obtain

\begin{dmath}
\label{covariance_X}
 \mathbb{E}(X^{(D,\gamma)}_t X^{(D,\gamma)}_s)=\frac{1}{2}\left(\frac{H}{\alpha}\right)^{2H}
\left( e^{\alpha(t-s)}+e^{-\alpha(t-s)}-e^{\alpha(t-s)}\left(1-e^{\frac{\alpha(t-s)}{H}}\right)^{2H}\right)
\end{dmath}

$Xt^{(D,\alpha)}$ is a stationary process. In particular, using (\ref{covariance_X}) and the self-similarity property of the fractional Brownian
motion, it may be proven that $X_t^{(D,\alpha)}$ is normally distributed with mean zero and variance $(H/\alpha)^{2H}$, for all $t$.

Consider next the process $Y^{\alpha}$ defined via
\begin{equation}
 Y_t^{(\alpha)} := \int^t_0 e^{-\alpha s} \,\mathrm{d}B^H_{a_s}
\end{equation}

The process $Y_t^{(\alpha)}$ has stationary increments. Using $Y^{(\alpha)}$ the process $ X^{(D,\alpha)}$ may be viewed as 
the solution of the equation
\begin{equation}
 \mathrm{d}X_t^{(D,\alpha)} = -\alpha X_t^{(D,\alpha)}\,\mathrm{d}t+\mathrm{d}Y_t^{(\alpha)},
\end{equation}
with random initial value $X_0^{(0,\alpha)}=B^H_{a_0}\stackrel{d}{=}B^H_{H/\alpha}\sim N(0, (\frac{H}{\alpha})^{2H})$.

We now consider the Langevin SDE with $Y^{(1)}$ as the driving process:
\begin{equation}
\label{Langevin}
 \mathrm{d}U_t^{(D,\gamma)} = -\gamma U_t^{(D,\gamma)}\,\mathrm{d}t+\mathrm{d}Y_t^{(1)},\;\;\gamma > 0,
\end{equation}
The solution can be expressed as 
\begin{equation}
\label{U}
 U_t^{(D,\gamma)}=e^{-\gamma t}\int^t_{-\infty} e^{\gamma s} \,\mathrm{d}\hat{Y}_s^{(1)}=
e^{-\gamma t}\int^t_{-\infty} e^{(\gamma-1) s} \,\mathrm{d}B^H_{a_s},\;\; \gamma > 0,
\end{equation}

where $\hat{Y}_s^{(1)}$ is the two sided $Y^{(1)}$ process and $\alpha=1$ in $a_t$.

\begin{mydef}
The process $U_t^{(D,\gamma)}$ defined in (\ref{U}) or, equivalently, via the SDE (\ref{Langevin}) is
called the fractional Ornstein-Uhlenbeck process of second kind ($fOU_2$) with initial value $B^H_{H/\alpha}$.
\end{mydef}


\begin{remark}
By Proposition 3.11 in \cite{Paavo}, the covariance of the process $U_t^{(D,\gamma)}$ decays exponentially and has short range dependence.
\end{remark}
\begin{remark}
The process $U^{(D,\gamma)}$ has quadratic variation zero.
\end{remark}

\begin{proof}
By proposition 3.4, \cite{Paavo}, the sample paths of $U^{(D,\gamma)}$ are H\"older of order $\beta$ for $\forall \beta<H$.
For $\frac{1}{2}<\beta<H$,
\begin{dmath*}
\left(U^{(D,\gamma)}_t-U^{(D,\gamma)}_s \right)^2 \leq  K_T(\omega) \left|t-s\right|^{2\beta}.
\end{dmath*}
Therefore, for any sequence $\pi_n$ of partitions of the interval $[0,T]$ such that $|\pi_n|\rightarrow 0$. 
\begin{dmath*}
\left[ U^{(D,\gamma)}, U^{(D,\gamma)} \right]_T=\mathbb{P}-\lim_{|\pi_n|\rightarrow 0} \sum_{t_k\in \pi_n} 
\left( U^{(D,\gamma)}_{t_k}-U^{(D,\gamma)}_{t_{k-1}}\right)^2 \leq K_T(\omega) \lim_{|\pi_n|\rightarrow 0} \sum_{t_k\in \pi_n}\left|t_k-t_{k-1}\right|^{2\beta}
\leq K_T(\omega) \lim_{|\pi_n|\rightarrow 0} |\pi_n|^{2\beta-1}\, T =0
\end{dmath*}
almost surely as n tends to infinity.
\end{proof}

\begin{prop}\cite{Paavo}
\label{covariance_OU}
The autocovariance of $U^{(D,\gamma)}$ has the kernel representation
\begin{equation*}
\mathbb{E}(U^{(D,\gamma)}_t U^{(D,\gamma)}_s)
=H(2H-1)H^{2H-2}e^{-\gamma(t+s)}\int_{-\infty}^t \int_{-\infty}^s\! \frac{e^{(\gamma-1+\frac{1}{H})(u+v)}}{\left|
e^{u/H}-e^{v/H}\right|^{2(1-H)}} \, \mathrm{d}u\, \mathrm{d}v
\end{equation*}
\end{prop}

We give another expression in discrete time for the autocovariance function which may be used for computational calculations. 
To simplify the notation, we write $\mathcal{Q}^{n,T}(k-j)$ at lag $(k-j)$ instead of $\mathcal{Q}^{n,T}(t_k-t_j)$.

\begin{prop}
\label{covariance_OU_discrete}
Consider the time interval $[0,T]$ and an equidistant partition $\Pi := \{ t_i=\frac{iT}{n};\,0\leq i\leq n\}$.
Let $t_j,t_k \in \Pi$ with $j\leq k$. The autocovariance of $U^{(D,\gamma)}$ at lag $(k-j)$ is 
\small
\begin{gather*}
\mathcal{Q}^{n,T}(k-j)=\mathbb{E}(U^{(D,\gamma)}_{t_{j+(k-j)}} U^{(D,\gamma)}_{t_j}) =
 H^{-2(\gamma-1)H}H\left(2H-1\right)e^{-\gamma(t_{j+(k-j)}-t_j))}\times\\
\biggl(\frac{H}{\gamma H} \mathbf{B}((\gamma-1)H+1,2H-1)+
\int_H ^{a_{t_{j+(k-j)}-t_j}}\! m^{2\gamma H-1}\mathbf{B}(H/m;(\gamma-1)H+1,2H-1) \,\mathrm{d}m\biggr)
\end{gather*} 
\normalsize
with $\mathbf{B}(\cdot,\cdot)$ for the beta function and $\mathbf{B}(\cdot\,;\cdot,\cdot)$ for the incomplete beta function with 
$\mathbf{B}(1\,;\cdot,\cdot)=\mathbf{B}(\cdot,\cdot)$.
\end{prop}

\begin{proof} [\textbf{Proof}] 
 See Appendix A. 
\end{proof}

\begin{coro}
\label{covariance_increment_OU}
Let denote the autocovariance of the increment process of $U^{(D,\gamma)}$ at lag $m$ by $\mathcal{C}^{n,T}(m)$. Then 
its autocovariance function takes the form:
\begin{equation}
\label{autocovariance_fOU2}
 \mathcal{C}^{n,T}(m)= 2\mathcal{Q}^{n,T}(m)-[\mathcal{Q}^{n,T}(m-1)+\mathcal{Q}^{n,T}(m+1)]
\end{equation}
with $0\leq m \leq n-1$. 
\end{coro}

\begin{proof} [\textbf{Proof}]
\begin{dmath*}
\mathcal{C}^{n,T}(m)=\mathbb{E}\left( \Delta U^{(D,\gamma)}_{t_{m+1}} \Delta U^{(D,\gamma)}_{t_1}\right)= \mathbb{E}\left[\left(U^{(D,\gamma)}_{t_{m+1}}
-U^{(D,\gamma)}_{t_m}\right) \left(U^{(D,\gamma)}_{t_1}-U^{(D,\gamma)}_{t_0}\right)\right]\\
=\mathbb{E}\left(U^{(D,\gamma)}_{t_{m+1}} U^{(D,\gamma)}_{t_1}\right)+\mathbb{E}\left(U^{(D,\gamma)}_{t_m} U^{(D,\gamma)}_{t_0}\right)
-\mathbb{E}\left(U^{(D,\gamma)}_{t_m} U^{(D,\gamma)}_{t_1}\right)\\
-\mathbb{E}\left(U^{(D,\gamma)}_{t_{m+1}} U^{(D,\gamma)}_{t_0}\right)= 2\mathcal{Q}^{n,T}(m)-[\mathcal{Q}^{n,T}(m-1)+\mathcal{Q}^{n,T}(m+1)]
\end{dmath*}
\end{proof}

\begin{remark}
We give two examples for proposition (\ref{covariance_OU_discrete}) and corollary (\ref{covariance_increment_OU}) respectively.
Autocovariance function for a process $U^{(D,0.3)}$ with parameter $H=0.9$, cf. Figure 1, Appendix B and 
autocovariance function for an increment process $U^{(D,0.3)}$ with parameter $H=0.9$, cf. Figure 2, Appendix B.
\end{remark}

\subsection{Mixed Fractional Gaussian process}

Now we are ready to construct a family of continuous processes $X$ which is Gaussian and it has the following properties

\begin{enumerate}
\item[(i)] Let $\varLambda=\{t_0,\ldots, t_n\}$ with $0=t_0<\ldots<t_n=t$, be a partition of $[0,t]$ 
and $\| \varLambda \|=\mathrm{max}_{1\leq k \leq n}|t_k-t_{k-1}|$. Then the quadratic variation process 
$\langle X \rangle_t=\lim_{\|\varLambda\|\to 0}\sum_{k=0}^n(X_{t_{k+1}}-X_{t_k})^2=t$. 
\item[(ii)] The corresponding increment process $\varXi_{t_{k}}=X_{t_{k}}-X_{t_{k-1}}$ has a similar autocovariance structure to 
that of a FARIMA, i.e, it captures the short and long range dependences.
\end{enumerate}
We construct this process as $X_t=\sigma W_t + B_t^H + U^{(D,\gamma)}_t$ with $H\in(\frac{1}{2},1)$ and $\gamma > 0$. 
We assume the three processes are mutually independent and the fractional Gaussian 
noise $B_t^H$ and the fractional Ornstein-Uhlenbeck process of the second kind $U^{(D,\gamma)}_t$ may have different Hurst 
parameter. Then, its increment process is defined as 

\begin{equation}
 \varXi_{t_k} = \sigma \Delta W_{t_k} + Z_{t_k}^H + \Delta U^{(D,\gamma)}_{t_K}
\end{equation}

Notice, that at first, the fractional Gaussian 
noise $Z_t^H$ and the increment fractional Ornstein-Uhlenbeck process of the second kind $\Delta U^{(D,\gamma)}_t$ may have different Hurst 
parameter. However, from a statistical view, .... 

The fractional Gaussian noise process $Z_t^H$ captures the long range dependencies and its autocovariance function behaves asymptotically
as a FARIMA. However, if the data contains strong short correlations, it fails to capture them, cf. Fig. \ref{Autocovariance_comparetion}.

To model the short range correlations, we add the increment fractional Ornstein-Uhlenbeck process of the second kind. As a result, we obtain a 
process with similar autocovariance structure as FARIMA or fractional exponential process.

In many applications in continuous time,
such as in a delta hedging problem, we need to use  It\^o's formula. However, we need to justify its use since 
$fBm$ and $fOU_2$ have both quadratic variation zero. Therefore, a increment Brownian motion $B_t$ is added so that the process $X$ has a 
continuous quadratic variation as $\langle X \rangle_t=\sigma^2 t$. Moreover, the 
increments of the Brownian process are independent so the autocovarinace function of $Z$ does not change. 
Then, by proposition \ref{mixed_variance}, the structure of its increment process $Z$ is similar to that of FARIMA process, see Figs. \ref{Autocovariance_comparetion} 
and \ref{Autocovariance_mixed_model_2}. The new process $Z$ is more parsimonious than that of a FARIMA with a
consequent reduction of errors in model estimation and forecasting.

\begin{prop}
\label{mixed_variance}
Define $\varXi$ as a mixed fractional Gaussian process, i.e., $\varXi_{t_k}=\sigma \Delta B_{t_k} + Z_{t_k}^H + \Delta U^{(D,\gamma)}_{t_k}$ with $H> \frac{1}{2}$ and $\gamma > 0$
in an interval $[0,T]$. Then the variance is computed as 
\begin{dmath*}
  \mathbb{E}\left(\varXi_{t_k}\right)^2 
= \mathbb{E}\left(W_{t_{k+1}}-W_{t_k}\right)^2 + \mathbb{E}\left(B^H_{t_{k+1}}-B^H_{t_k}\right)^2 + \mathbb{E}\left(U^{(D,\gamma)}_{t_{k+1}}-U^{(D,\gamma)}_{t_k}\right)^2
\end{dmath*}
\end{prop}

\vspace{5pt}
\begin{remark}
 Because $\varXi$ is stationary its variance may be written in terms of the first increment as 
\begin{dmath*}
  \mathbb{E}\left(\varXi_{t_0} \right)^2 
= \sigma(t_1-t_0)+(t_1-t_0)^{2H}+2\left(\mathcal{Q}^{n,T}(0)-\mathcal{Q}^{n,T}(1)\right) 
= \sigma t_1+t_1^{2H}+2\mathbb{E}\left(U^{(D,\gamma)}_{t_0}\right)^2-2\mathbb{E}\left(U^{(D,\gamma)}_{t_1}U^{(D,\gamma)}_{t_0}\right).
\end{dmath*}
\end{remark}

By construction, we can apply It\^o's formula to our new process $X$.

We give a definition of foward integral due to \cite{Bender1}.
\begin{mydef}
Let $t\leq T$ and $X:\mathbb{R}^+ \longrightarrow \mathbb{R}$ 
be a continuous process. The forward integral of a process
$Y$ with respect to $X$ along equidistant  $\pi_n$ partition of the interval $[0,T]$ such that $|\pi_n|\rightarrow 0$ is
\begin {equation}
 \int _0^t\! Y_s\,\mathrm{d}X_s := \lim_{n\rightarrow \infty} \sum_{k=1}^{n} Y_{t_k} (X_{t_{k+1}}-X_{t_k}) \quad \mathrm{with} \quad t_k\in \pi_n
\end {equation}
when the  $\mathbb{P}$-a.s limit exits.
\end{mydef}

If $X$ is a continuous process with continuous quadratic variation 
$\langle X \rangle_t$ such that $\langle X \rangle_t=\lim_{\|\varLambda\|\to 0}\sum_{k=0}^n(X_{t_{k+1}}-X_{t_k})^2$ $\mathbb{P}$-a.s then
we have the following It\^o's formula according to \cite{Sondermann}

\begin{theorem}
Let $X:[0,\infty)\longrightarrow\mathbb{R}^1$ be a continuous function with continous quadratic variation $\langle X \rangle_t$,
and $f \in \mathcal{C}^{1,2}([0,t]\times \mathbb{R})$ a twice differenciable real function. Then 
\begin{equation*}
f(X_t,t)= f(X_0,0)+\int _{0}^t\! f_s(X_s,s)\,\mathrm{d}s+
 \int _0^t\! f_x(X_s,s)\,\mathrm{d}X_s+\frac{1}{2} \int _0^t\! f_{xx}(X_s,s)\,\mathrm{d}\langle X \rangle_s
\end{equation*}
for any $t\geq 0$.
\end{theorem}

\begin{remark}
 Note that the integral $\int _{0}^t\! f_x(X_s,s)\,\mathrm{d}X_s$ is understood as a forward integral along the partition $\pi_n$.
\end{remark}

\section{Temporal aggregation}

We shortly present the impact of temporal aggregation on a FARIMA process.
\begin{mydef}
 Let $n=mT$ with $m\geq 2$ and $L$ the lag operator, then the series 
\begin{equation}
 Y_T=\left( \sum_{i=0}^{m-1} L^i\right)y_{mT}
\end{equation}
 represents the m-period nonoverlapping aggregates of $y_n$.
\end{mydef}

A FARIMA(p,d,q) process follows the equation

\begin{equation}
\label{farima_model}
 \Phi_p(L)(1-L)^d y_n=\Theta_q(L)\varepsilon_n
\end{equation}

The original process and the aggregated one are linked via a polynomial. 
We multiply both sides of equation (\ref{farima_model}) by the polynomial

\begin{equation*}
\prod_{j=1}^p \left[ \frac{(1-\delta^m L^m)}{(1-\delta L^m)}\right]\left(\frac{1-B^m}{1-B}\right)^{d+1}
\end{equation*}
    
As a result, the aggregate series $Y_T$ follows a FARIMA(p, d, N) with
\begin{equation}
 N=\left( p+d+1+\frac{q-p-d-1}{m} \right)
\end{equation}

and autocovariance function as 

\begin{equation}
 \gamma_Y(j)=\left( \sum_{i=0}^{m-1} L^i \right)^{2d+1}\gamma_y(mj+(d+1)(m-1))
\end{equation}
for further details cf.\cite{Tesler, Wei, Stram}.

Conversely, the mixed fractional Gaussian noise has variance and autocovariance depending 
on the lenght interval $T$ and the sampling frequency $n$ as 

\begin{dmath*}
  \mathbb{E}\left(\varXi_{\frac{kT}{n}} \right)^2  
= \sigma \frac{T}{n}+\left(\frac{T}{n}\right)^{2H}+\mathrm{const}
\end{dmath*}
 
and

\begin{dmath*}
  \mathbb{E}\left(\varXi_{\frac{kT}{n}}\,\varXi_{\frac{jT}{n}} \right)^2  = \frac{1}{2}
\left(\frac{T}{n}\right)^{2H}\left(|(k-j)+1|^{2H}-2|k-j|^{2H}+|(k-j)-1|^{2H}\right)+\mathcal{C}^{n,T}(k-j)
\end{dmath*}

with $H>\frac{1}{2}$ respectively.

With finite length aggregation, the autocovariance structure of the aggregates would depend on the exact autocovariance
structure of the mixed fractional Gaussian noise.

\section{Hedging and Expected Shortfall for Options}

The need to quantify risk arises in many different contexts and has been strongly motivated by the 
fear of systemic risk, i.e. the danger that problems in a single financial institution 
may spill over and, in extreme situations, disrupt the normal functioning of the entire 
financial system.

Lessons learned from the global banking crisis are now spotlight in a review of risk management 
at all levels within financial institutions and regulatory authorities. Solvency II, Europe’s risk-based reform of 
insurance regulation, and Basel III, a global regulatory framework for banks on capital adequacy, leverage ratios 
and liquidity standards, will fundamentally shift the focus of the financial industry for many years to come. 
A central issue is the measurement of risk. Among the existing approches, Basel III mentions Value-at-Risk
(VaR) for raising capital requirements for the trading book and complex securitisation exposures,and Solvency II 
the related notion of expected shortfall (ES) is used in the definition of target capital. ES captures the skewed and
heavy-tailed pay-off functions. 

Calculation of VaR and ES essentially consist of determining the loss distribution function $F_X(x)=P(X\leq x)$ , 
or functionals describing this distribution function such as its mean, and variance. In order to achive this, 
a proper calibrated model is needed which captures the main features of the dynamics of the value of a financial potfolio.

The effects of driving stochastic processes mixed fGn and FARIMA are compared on the forecast of risk measures (VaR, ES) of 
a financial position.

\subsection{Characterization of risk measures}
In this paper, we pay attention to those measures applied in the framework of Basel III and Solvency II.
\begin{mydef}
For a financial position $X$ with distribution $\mathbb{P}$, we define its Value-at-Risk at level $\alpha$ $(VaR^{\alpha})$ as
\begin{equation}
 VaR^\alpha(X):=-q_X (\alpha) = \mathrm{inf}\{ m \;|\; \mathbb{P}(X\leq m) \geq \alpha\}
\end{equation}
\end{mydef}

where $q_X(\alpha)$ is the quantile function of $X$. From a point of a practitioner, $VaR$ is the maximum 
loss he may expect over a given holding or horizon period, with a certain level of confidence.

However, the subadditivity property fails to hold for $VaR$ in general, so $VaR$ is not a
coherent risk measure. For subadditivity mesure diversification always leads to risk reduction, while for measures which fail
this condition, diversification may increase in their value; cf. \cite{Acerbi}.

One posibility of a coherent measure which is defined in terms of $VaR$ would be Conditional Value at Risk or Expected Shortfall.  
\begin{mydef}
 Let $X$ be the financial position on a specified time horizon $T$ and some specified probability level 
$\alpha\in (0,1)$. The Excepted Shortfall is then defiend as  
\begin{equation}
ES^\alpha = \frac{1}{\alpha}\int^\alpha_0 VaR^p(X) \; \mathrm{d}p
\end{equation}
\end{mydef}

\begin{remark}
 If the distribution function of $X$ is continuous then it can be shown that $ES^{\alpha}= \mathbb{E}(X|X\leq VaR^\alpha)$.
\end{remark}

\subsection{Hedging and elimination of randomness}

The dynamics of the price of an underlying asset of a derivative product may be modelled according to a mixed fractional 
Gaussian process as

\begin{equation}\label{underlying}
dS_t=\mu S_{t}dt+\sigma S_{t}dX_t 
\end{equation}

Therefore, the dynamics of the option price is a function $F(S_t,t)\in \mathcal{C}^{1,2}([0,t]\times\mathbb{R})$ and according to It\^o's formula

\begin{dmath}\label{Ito}
F(S_t,t)=\frac{\partial F}{\partial t}dt+\frac{\partial F}{\partial S_t}\left(\mu S_{t}dt+\sigma S_{t}dX_t\right)+\frac{1}{2} \sigma^2 S_{t}^2 \frac{\partial^2 F}{\partial S_t^2}dt\\
= \sigma S_{t}\frac{\partial F}{\partial S_t}dX_t+\left(\mu S_{t}\frac{\partial F}{\partial S_t}+\frac{1}{2} \sigma^2 S_{t}^2\frac{\partial^2 F}{\partial S_t^2}+\frac{\partial F}{\partial t}\right)dt
\end{dmath}

Let $\xi_t$ represent the value of a portfolio of one option with value $C(S_t,t)$ and $-\eta_t$ underlying 
stocks with price $S_t$. 
The minus sign of $\eta_t$ means we hold a short position in the underlying asset. Therefore, the value of the portfolio 
at time $t$ is $\xi_t=C_t-\eta_t S_t$. 
We can write
\begin{dmath*}
 d\xi_t=dC_t-\eta_t\, dS_t = \sigma S_{t}\frac{\partial C}{\partial S_t}dX_t+\left(\mu S_{t}\frac{\partial C}{\partial S_t}+\frac{1}{2} \sigma^2 S_{t}^2 \frac{\partial^2 C}{\partial S_t^2}+\frac{\partial C}{\partial t}\right)dt
-\eta_t\left(\mu S_{t}dt+\sigma S_{t}dX_t \right)
=  \sigma S_{t}\left(\frac{\partial C}{\partial S_t}-\eta_t\right)dX_t+\left( \mu S_{t}\left(\frac{\partial C}{\partial S_t}-\eta_t\right)+\frac{1}{2} \sigma^2 S_{t}^2\frac{\partial^2 C}{\partial S_t^2}+\frac{\partial C}{\partial t}\right)dt
\end{dmath*}

where we have substituted equations (\ref{Ito}) and (\ref{underlying}) into $dC(S_t,t)$ and $dS_t$ respectively. 
Now, if

\begin{equation*}
 \eta_t=\frac{\partial C(S_t,t)}{\partial S_t}
\end{equation*}
we eliminate the randonmess of the porfolio and by fact that $C(S_t,t)$ and $S_t$ are correlated 
implies that the option price will change by
\begin{equation}\label{delta}
 dC_t = \eta_t  \,dS_t
\end{equation}
respect to the underlying price at time $t$ \cite{Wilmott}.
 
\begin{ex}

Suppose we have a portfolio with one option and one stock with value $\xi_t=C_t+S_t$. Using 
the delta approximation (\ref{delta}) its value is a linear function of $S_t$ alone as $\xi_t=(\eta_t + 1) S_t$ and any change 
is given by $d\xi_t=(\eta_t + 1) dS_t$ at any $t$.
We assume the distribution of the returns for the stock to be normal,
\begin{equation*}
 r_{u,t}\sim N(\mu_{u,t},\sigma_{u,t}^2)
\end{equation*}
where $\mu_{u,t}$ is the conditional mean calculated according to equation (\ref{cond_mean}) and $\sigma_{u,t}^2$ 
is the conditional variance calculated by equation (\ref{cond_variance}). Therefore, the distribution of the returns of the 
portfolio is also normally distributed as  

\begin{equation*}
 r_{p,t}\sim (\eta_t +1)N\left(\mu_{u,t},\sigma_{u,t}^2\right)
\end{equation*}

Let us denote the value at risk on the underlying by $VaR_u^\alpha$, where $\alpha$ is the confidence level, then

\begin{equation}
VaR_u^\alpha=\mu_{u,t}S_{t-1}+\sigma_{u,t} \Phi^{-1}(\alpha)S_{t-1}
\end{equation}
Let us denote the value at risk on the portfolio by $VaR_p^{\alpha}$, where $\alpha$ is probability and recall that Var is a quantile of the loss distribution of the portfolio then:
\begin{eqnarray*}
 \alpha&=&\mathbb{P}(\xi_t-\xi_{t-1}\leq VaR_p^\alpha)\\
&=&\mathbb{P}((\eta_t+1)\,(S_t-S_{t-1})\leq VaR_p^\alpha)\\
&=&\mathbb{P}((\eta_t+1)S_{t-1}(e^{r_{u,t}}-1)\leq VaR_p^\alpha)\\
&=&\mathbb{P}((\eta_t+1)S_{t-1}(e^{r_{u,t}}-1)\leq VaR_p^\alpha)\\
&=&\mathbb{P}\left(\frac{r_{u,t}-\mu_{u,t}}{\sigma_{u,t}}\leq \left[\log\left(\frac{1}{(\eta_t+1)}\frac{VaR_p^\alpha}{S_{t-1}}+1\right)-\mu_{u,t}\right]\frac{1} {\sigma_{u,t}}\right)\\
\end{eqnarray*}
Using the normality assumption of returns
\begin{equation*}
\Phi^{-1}(\alpha)=\left[\log\left(\frac{1}{(\eta_t+1)}\frac{VaR_p^\alpha}{S_{t-1}}+1\right)-\mu_{u,t}\right]\frac{1} {\sigma_{u,t}}
\end{equation*}

for small $\frac{1}{(\eta_t+1)}\frac{VaR_p^\alpha}{S_{t-1}}$ we can use the approximation $log(1+x)\approx x$.

Hence, the value at risk for one unit of the portfolio at confidence level $\alpha$ is:
\begin{equation*}
 VaR_p^\alpha = (\eta_t+1)\, \sigma_{u,t} \Phi^{-1}(p)S_{t-1}+\mu_{u,t} (\eta_t+1)\,S_{t-1} =(\eta_t+1)VaR_u^\alpha
\end{equation*}

If we assume that the underlying stock log price is modelled by a mixed fGn process (\ref{underlying}) then
the excepted shortfall for the portfolio at level $\alpha$ is 

\begin{equation}
 ES_p^\alpha=-\left(\mu_{u,t}+\frac{1}{\alpha}\sigma_{u,t}\,\phi\left((\eta_t+1)VaR_{u}^\alpha \right)\right)\xi_{t-1}.
\end{equation}

Now, $\sigma_{u,t}$, which is a function of Hurst and gamma parameter, 
can be evaluated via proposition (\ref{mixed_variance}). Note that we used the result that mixed $fGn$ has continuous quadratic 
variation equal to $\langle B \rangle_t$ so we were able to use Ito's formula to 
justifie the linear approximation of the increment value of the option.

\end{ex}

\section{Backtesting study}

We check if mixed models are good altervative to Farima models to forecast risk of financial data which exhibits 
short-long range dependences. 

\subsection{Prediction of conditional mean and variance at time $n+k$}

We assume our random variables are jointly gaussian. We denote the best linear predictor of $X_{n+k}$ 
as $\widehat{X}_{n+k}=\sum_{i=1}^n a_{i,k} X_i$ and use the mean squared error (MSE) as our criterio, 
$\lVert X_{n+k}-\widehat{X}_{n+k}\rVert _{L^2}= E((X_{n+k}-\widehat{X}_{n+k})^2)$. 
Assuming that the process is weakly stationary, let $\widehat{x}_{n+k}$ denote the minimum mean square error linear predictor of 
$x_{n+k}$ given the data $\bar{X}'=(x_1,\ldots,x_n)$, 
the mean $\mu$ and the autocovariances $\gamma_l$, with $l=0,\ldots,n-1$. 

\begin{equation}
\label{cond_mean}
 \widehat{x}_{n+k}=\mu+g_k'\Gamma^{-1}_n(\bar{x}-\mu)
\end{equation}

where $g'_k=(\gamma_{n+k-1},\ldots,\gamma_k)$ and, by the law of total variance, 
the conditional variance for the forecast is

\begin{equation}
\label{cond_variance}
 V_{n+k}=\gamma_0-g'_k\Gamma^{-1}_n g_k
\end{equation}

We compute the preditors of mean and variance by means of  Durbin-Levinson algorithm \cite{Golub}.

\subsection{Comparation of risk model performance.}

Assement of the accuracy of the expected shortfall forecasts should ideally be done by 
monitoring the performance of the model in the future. However, it is expected that violations are only observed 
rarely and a long period of time would be required. Backtersting is a procedure used to compare risk model performances 
over a period in the past.

In our study, we are not concern with the estimation of the parameters of the models but to compare their performances.
Therefore, we simulate data from a FARIMA model. This is two fold, first it allows us to control the 
dependences of the data and second use the FARIMA forecast of risk as a benchmark to evaluate the performance of the that
of the mixed model.

We assume the model parameters are fixed except the gamma parameter which is calibrated so the predicted conditional 
variance of the mixed model approximates that of the FARIMA. We assume that the data is independent. This is questionable
assumption since we are concern with the correlation in the data but we may get an inside of the validation of the modeling
as a first approximation. Future reasearch with more formal test of violation ratios 
would be need to obtain a better conclusion.

We processed with the calculation of the autocovariance functions according to equations (\ref{autocovariance_farima}),
(\ref{autocovariance_fOU2}) and proposition (2.1-4). The Expected Shortfall and Value at risk is then evaluated
using section (2.1-2) of the general theory.


We analyze the results by means of violation ratios and volatility. If the return on a particular day exceeds the 
forecast, then we count it as a violation. Let $\varsigma$ be a Bernoulli random variable with probability
the risk level of the ES, where $\varsigma_1$ is the 
number of violations and $\varsigma_0$ is the number of non violations then the violation ratio is:

\begin{equation}
\varPsi = \frac{\varsigma_1}{\mathbb{E}(\varsigma)}
\end{equation}

As a result, numerical results are presented in the next Table

\begin{table}[ht]
\begin{center}
\begin{tabular}{ccc}
\toprule
Model & Ratio & Ratio volatility \\
\hline
Farima & 3.6 & 9.2e-06 \\
fGn & 4.0 & 8.6e-06 \\
Mixed & 3.6 & 9.2-06 \\
\bottomrule
\end{tabular}
\end{center}
\end{table}

None of the models perform well...

\section{Appendix A - Calculation of $fOU_2$ autocovariance}

We start the calculation of the $fOU_2$ kernel representation from proposition (\ref{covariance_OU}).

Recall that the process $U^{(D,\gamma)}$ was defined in (\ref{U}) as 
\begin{eqnarray*}
U_t^{(D,\gamma)}&=& e^{-\gamma t} \int_{-\infty}^t \! e^{(\gamma-1)s} \,\mathrm{d}Z_{a_s}\\
&=& H^{-(\gamma-1)H}e^{-\gamma t}\int_0^{a_t} \!s^{(\gamma-1)H}\,\mathrm{d}Z_s
\end{eqnarray*}
where $a_t:= a(t,H):=He^{t /H}$ and $\gamma > 0$. A change of variable was made as $s=He^{s/H}$.

To calculate the integral, we start defining the constant $C_1\equiv H(2H-1)$ then
\begin{dmath*}
C_1e^{-\gamma(t+s)}\int_{-\infty}^t \int_{-\infty}^s\! H^{2(H-1)}\frac{e^{(\gamma-1+\frac{1}{H})(u+v)}}{\left|
e^{u/H}-e^{v/H}\right|^{2(1-H)}} \, \mathrm{d}u\, \mathrm{d}v \\
=C_1e^{-\gamma(t+s)}\int_{-\infty}^t \int_{-\infty}^s\!e^{(\gamma-1)(u+v)}e^{\frac{1}{H}(u+v)}\left(H\left|
e^{u/H}-e^{v/H}\right|\right)^{2(H-1)}\,\mathrm{d}u\,\mathrm{d}v \\
=H^{-2(\gamma-1)H}C_1e^{-\gamma(t+s)}\int_{-\infty}^t \int_{-\infty}^s\!H^{2(\gamma-1)H}e^{(\gamma-1)\frac{(u+v)H}{H}}e^{\frac{u}{H}}e^{\frac{v}{H}}
\times \left(\left|He^{u/H}-He^{v/H}\right|\right)^{2(H-1)}\,\mathrm{d}u\,\mathrm{d}v \\
=H^{-2(\gamma-1)H}C_1e^{-\gamma(t+s)}\int_{-\infty}^t \int_{-\infty}^s\!\left(He^{\frac{u}{H}}He^{\frac{v}{H}}\right)^{(\gamma-1)H}e^{\frac{u}{H}}e^{\frac{v}{H}}
\times \left(\left|He^{u/H}-He^{v/H}\right|\right)^{2(H-1)}\,\mathrm{d}u\,\mathrm{d}v \\
\end{dmath*}

Next, we make a change of variable $m=He^{u/H}$ and $n=He^{v/H}$. The constant term is now
 $C_2=H^{-2(\gamma-1)H}C_1$. Hence,
\begin{dmath*}
 =C_2e^{-\gamma(t+s)}\int_{0}^{a_t} \int_{0}^{a_s}\! (mn)^{(\gamma-1)H} |m-n|^{2(H-1)} \,\mathrm{d}m\,\mathrm{d}n \\
 = C_2e^{-\gamma(t+s)} \left(\int_{0}^{a_s} \int_{0}^{a_s}\! (mn)^{(\gamma-1)H} |m-n|^{2(H-1)} 
\,\mathrm{d}m\,\mathrm{d}n + \int_{a_s}^{a_t} \int_{0}^{a_s}\! (mn)^{(\gamma-1)H} |m-n|^{2(H-1)} 
\,\mathrm{d}m\,\mathrm{d}n \right)\\
= C_2e^{-\gamma(t+s)} \left(2 \int_{0}^{a_s} \int_{0}^{m}\! (mn)^{(\gamma-1)H} |m-n|^{2(H-1)} 
\,\mathrm{d}m\,\mathrm{d}n + \int_{a_s}^{a_t} \int_{0}^{a_s}\! (mn)^{(\gamma-1)H} |m-n|^{2(H-1)} 
\,\mathrm{d}m\,\mathrm{d}n \right)\\
\end{dmath*}

We continue by making a second change of variable, $\theta=\frac{n}{m}$, with the result

\begin{dmath*}
 =C_2e^{-\gamma(t+s)}\left(2\int_{0}^{a_s}\! m^{2\gamma H-1}\int_{0}^{1}\! \theta^{(\gamma-1)H} |1-\theta|^{2(H-1)}\,\mathrm{d}\theta\,\mathrm{d}m
+ \int_{a_s}^{a_t}\! m^{2\gamma H-1}\int_{0}^{a_s/m}\! \theta^{(\gamma-1)H} |1-\theta|^{2(H-1)}\,\mathrm{d}\theta\,\mathrm{d}m\right) \\
=C_2e^{-\gamma(t+s)}\left(2\; \mathbf{B}((\gamma-1)H+1,2H-1)\int_{0}^{a_s}\! m^{2\gamma H-1}\,\mathrm{d}m
+ \int_{a_s}^{a_t}\! m^{2\gamma H-1}\mathbf{B}\left(a_s/m;(\gamma-1)H+1,2H-1 \right) \,\mathrm{d}m\right) \\
\end{dmath*}

Finally, we obtain the desire result as
\small
\begin{adjustwidth}{-0.1in}{-0.5in}
\begin{dmath*}
\mathbb{E}(U^{(D,\gamma)}_t,U^{(D,\gamma)}_s)= H^{-2(\gamma-1)H}H(2H-1)e^{-\gamma(t+s)}\biggl(\frac{a_s^{2\gamma H}}{\gamma H} \mathbf{B}((\gamma-1)H+1,2H-1)\\
+\int_{a_s}^{a_t}\! m^{2\gamma H-1}\mathbf{B}(a_s/m;(\gamma-1)H+1,2H-1) \,\mathrm{d}m\biggr)\\
\end{dmath*}
\end{adjustwidth}
\normalsize

\section{Appendix B - Figures}

\begin{figure}[H]
\begin{center}
\includegraphics[width=\textwidth]{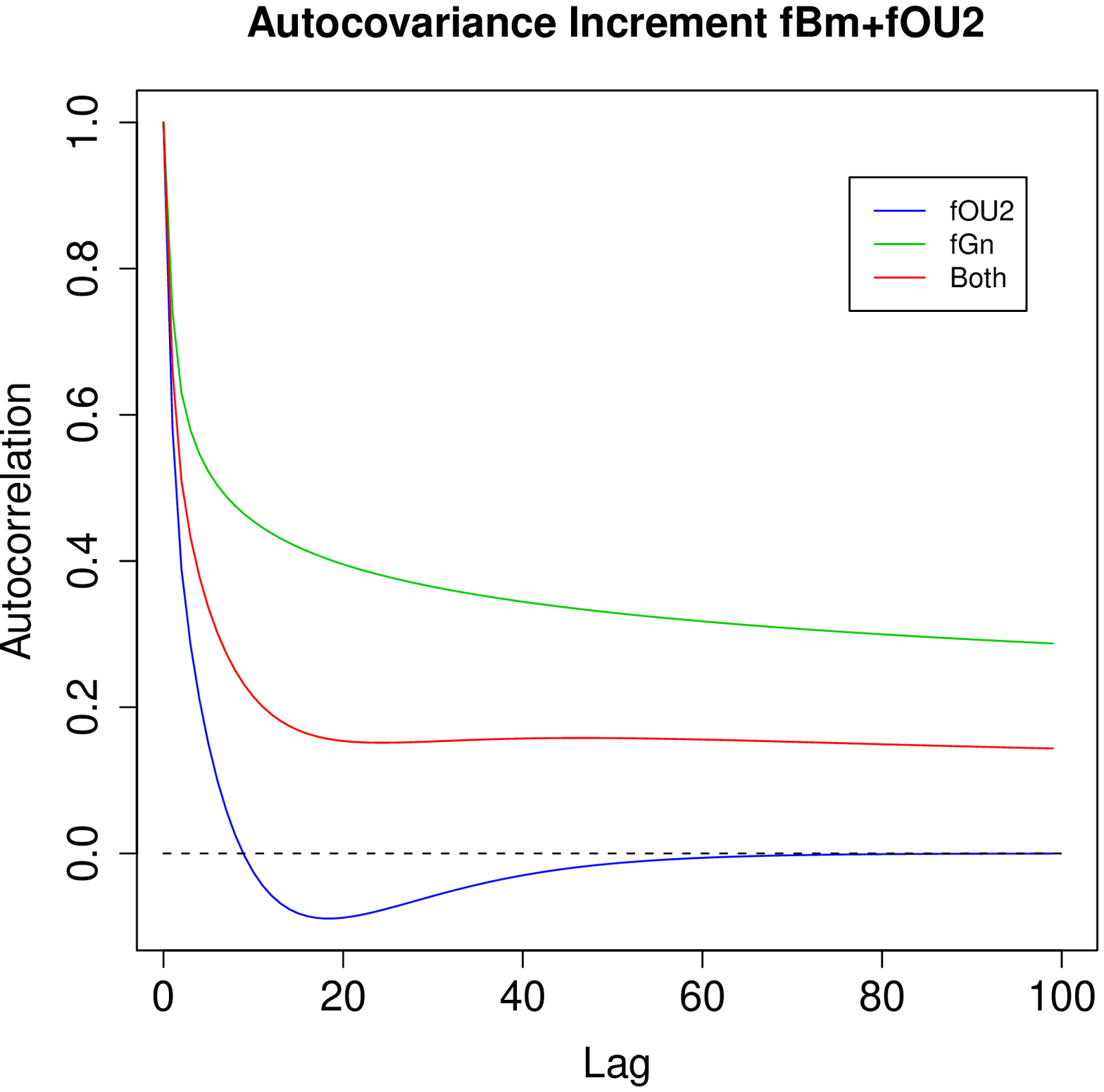}
\caption{\small Autocovariance effect of a fractional Ornstein-Uhlenbeck 
process of second kind $U_t^{(D,\gamma)}$ with parameters $\gamma=0.1$ and $H=0.9$ in the autocovariance function of
a fractional Gaussian noise with parameter $H=0.9$}.
\end{center}
\label{Autocovariance_mixed_model_1}
\end{figure}

\begin{figure}[H]
\begin{center}
\includegraphics[width=0.5\textwidth]{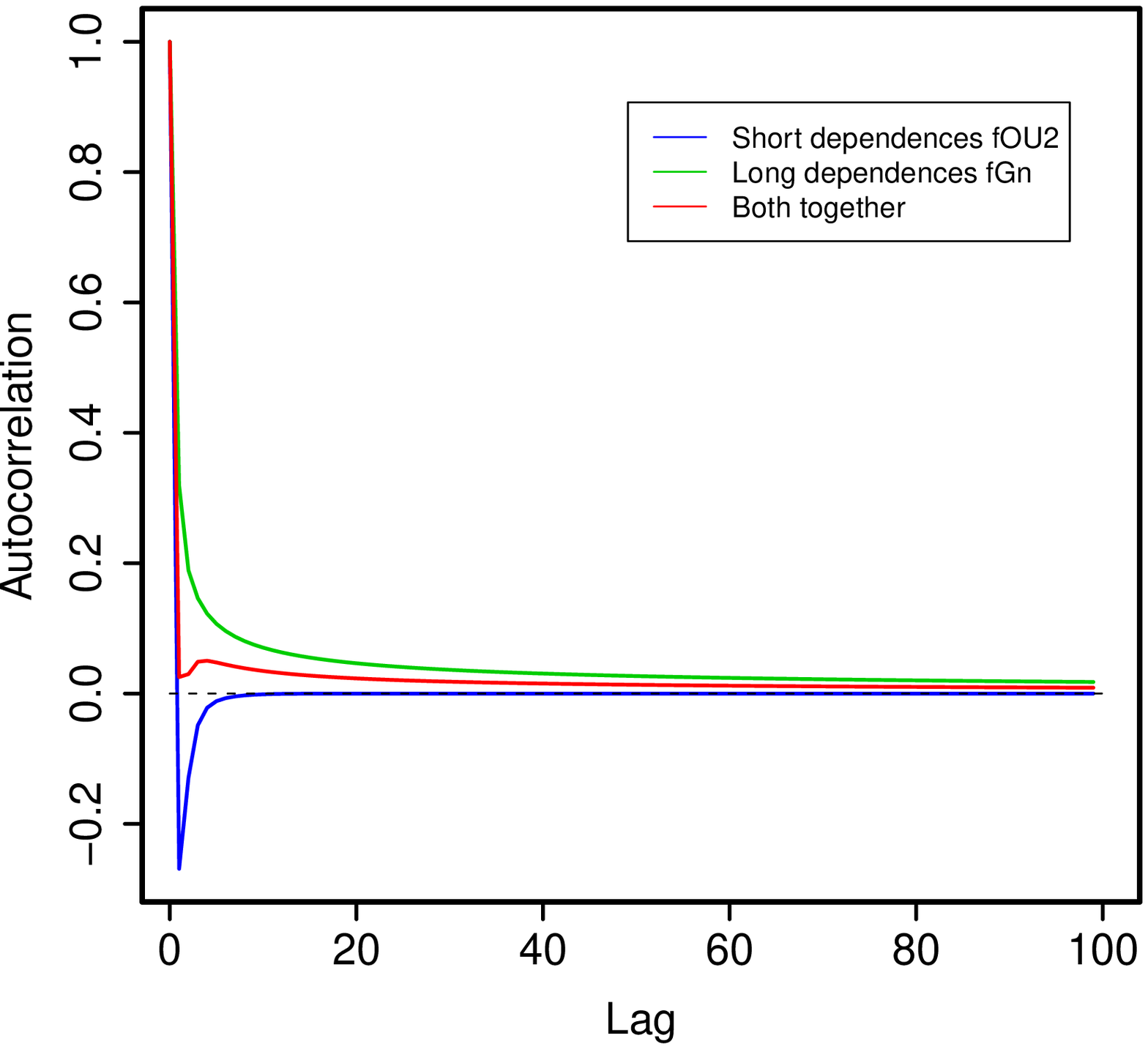}
\caption{\small Autocovariance effect of a fractional Ornstein-Uhlenbeck 
process of second kind $U_t^{(D,\gamma)}$ with parameters $\gamma=1.5$ and $H=0.7$ in the autocovariance function of
a fractional Gaussian noise with parameter $H=0.7$}.
\end{center}
\label{Autocovariance_mixed_model_2}
\end{figure}

\begin{figure}[H]
\label{Autocovariances_fOU2}
\begin{center}
\includegraphics[width=1\textwidth,height=0.3\textheight]{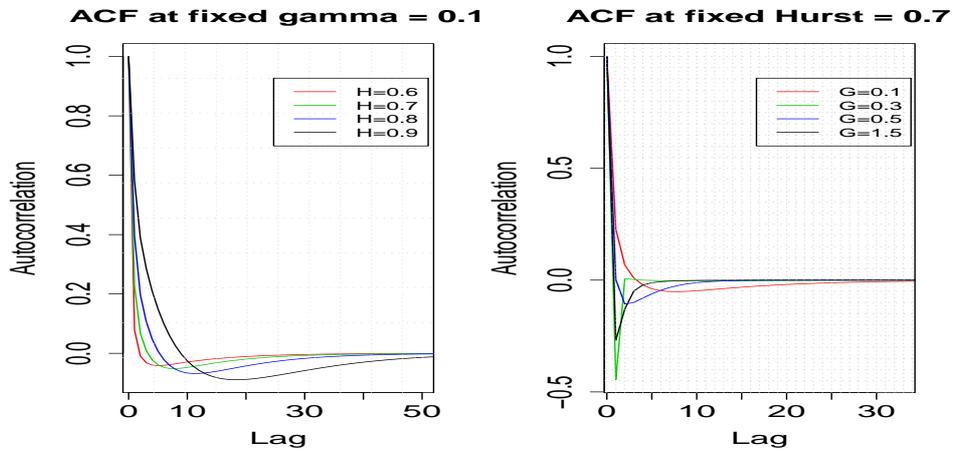}
\caption{\small In the left picture, autocovariance functions of an increment fractional Ornstein-Uhlenbeck 
process of second kind $U_t^{(D,\gamma)}$ at different $H$ and fixed $\gamma = 0.1$. The second picture shows
the autocovariances of the same process with fixed Hurst parameter $H = 0.7$ at different gammas G.} 
\end{center}
\end{figure}

\addcontentsline{toc}{chapter}{\numberline{}Bibliography}


\begin{thebibliography}{10}

\bibitem{Acerbi} 
\textsc{Acerbi C.,Tasche D.} (2001)
\newblock Expected Shortfall: a natural coherent alternative to Value at Risk.
\newblock \emph{Quantitative Finance Papers from arXiv.org}.

\bibitem{Alexander} 
\textsc{Alexander J.M., Rudiger F., Embrechts P.} (2005)
\newblock Quantitative Risk Management.Concepts, Techniques and Tools.
\newblock \emph{Princeton University Press}.

\bibitem{Bender1} 
\textsc{Bender, C., Sottinen, T. and Valkeila, E.} (2000)
\newblock Fractional Processes as Models in Stochastic Finance.
\newblock \emph{Advanced Mathematical Methods for Finance}, Series in Mathematical Finance, Springer. 

\bibitem{Bender} 
\textsc{Bender, C., Sottinen, T. and Valkeila, E.} (2008)
\newblock Pricing by hedging and no-arbitrage beyond semimartingales.
\newblock \emph{Finance and Stochastics} Vol. 12, pages 441-468. 


\bibitem{Doornik} 
\textsc{Doornik J.A., Oomus M.} (2003)
\newblock Computational aspects of Maximum Likelihood Estimation of Autoregressive Fractionally Integrated Moving Average models.
\newblock \emph{Computational Statistics \& Data Analysis}, Elsevier, Vol. 42(3), pages 333-348.

\bibitem{Follmer} 
\textsc{F\"ollmer H., Schied A.} (2005)
\newblock Stochastic Finance. An Introduction in Discrete time. 2nd Edition.
\newblock \emph{de Gruyter Studies in Mathematics 27}.

\bibitem{Golub} 
\textsc{Golub G., Loan C.V.} (1996)
\newblock Matrix Computations.
\newblock \emph{John Hoptkins University Press, Baltimore.}

\bibitem{Hamilton} 
\textsc{Hamilton J.D.} (1994)
\newblock Time Series Analysis.
\newblock \emph{Princeton University Press, New Jersey.}

\bibitem{Karatzas} 
\textsc{Karatzas I., Shreve E.} (1998)
\newblock Brownian Motion and Stochastic Calculus. 
\newblock \emph{Graduate Texts in Mathematics}, Second edition, Springer.

\bibitem{Lo} 
\textsc{Lo A.} (1991)
\newblock Long-memory in stock market prices. 
\newblock \emph{Econometrica}, Vol 59, pages 1279-1313.

\bibitem{Paavo} 
\textsc{Kaarakka T., Salminen P.} (2011)
\newblock On Fractional Ornstein-Uhlenbeck process.
\newblock \emph{Communications on Stochastic Analysis.} Vol. 5, No.1 pages 121-133.

\bibitem{Mishura} 
\textsc{Mishura Y.} (2008)
\newblock Stochastic Calculus for Fractional Brownian Motion and Related Process.
\newblock \emph{Lecture notes in Mathematics}. Springer

\bibitem{Sondermann} 
\textsc{Sondermann D.} (2006)
\newblock Introduction to Stochastic Calculus for Finance: A new Didactic Approach.
\newblock \emph{Lecture Notes in Economics and Mathematical Systems 579. Springer.}

\bibitem{Sowell} 
\textsc{Sowell F.} (1992)
\newblock Maximum likelihood estimation of stationary univariate fractionally integrated time series models.
\newblock \emph{Journal of Econometrics.} Vol. 53, pages 165-188.

\bibitem{Stram} 
\textsc{Stram D.O, Wei W.W.S} (1986)
\newblock Temporal Aggregation in the ARIMA process.
\newblock \emph{Journal of Time Series analysis.} Vol. 7, pages 293-302.

\bibitem{Taqqu} 
\textsc{Taqqu M. S.} (1999)
\newblock Fractional Brownian Motion and Long-Range Dependence
\newblock \emph{Theory and Applications of Long-Range Dependence}

\bibitem{Tesler}
\textsc{Tesler L. G.} (1967)
\newblock Discrete Samples and Sums in Stationary Stochastic Processes
\newblock \emph{Journal of the American Statistical Association.} Vol. 62, pages 484-499.

\bibitem{Wei}
\textsc{Wei W.W.S.} (2006)
\newblock Time Series Analysis: Univariate and Multivariate methods.
\newblock \emph{Addison Wesley}.

\bibitem{Willinger} 
\textsc{Willinger W., Taqqu M., Teverosky V.} (1999)
\newblock Stock Market prices and long range dependences.
\newblock \emph{Finance and stochastics.} Vol. 3, pages 1-13. 



\bibitem{Wilmott} 
\textsc{Wilmott P. ,Howison S., Dewynne J.} (1995)
\newblock The Mathematics of Financial Derivatives.
\newblock \emph{Cambridge University Press.} 



\end{thebibliography}
\end{document}